\documentclass[runningheads]{llncs}  

\usepackage{microtype}
\usepackage{graphics}
\graphicspath{{./figures/},{./figures/}}

\usepackage{authblk}

\usepackage{graphicx}
\usepackage{amsmath}
\usepackage{enumerate}
\usepackage{color}
\usepackage{xspace}
\usepackage{url}

\usepackage{amsfonts}\usepackage{amssymb}
\usepackage{thmtools}\usepackage{mathscinet}
\usepackage{thm-restate}
    	\definecolor{darkgreen}{rgb}{0.01, 0.93, 0.29}
\definecolor{lightbrown}{rgb}{0.91, 0.4, 0.11}
\usepackage{framed}
\usepackage{bbding}

\newcommand{\ER}{\ensuremath{\exists \mathbb{R}}\xspace}
\newcommand{\etrinv}{\ensuremath{\textrm{ETR-INV}}\xspace}
\newcommand{\planetrinv}{\ensuremath{\textrm{Planar-ETR-INV}^{*}}\xspace}
\newcommand{\planetrinvOld}{\ensuremath{\textrm{Planar-ETR-INV}^{}}\xspace}
\def\ETR{\textsc{ETR}\xspace}
\newcommand{\DGP}{\ensuremath{\textsc{Graph in Polygon}}\xspace}
\newcommand{\PDE}{\ensuremath{\textsc{Partial  Drawing Extensibility}}\xspace}
\newcommand{\AGP}{\ensuremath{\textsc{Art Gallery Problem}}\xspace}
\newcommand{\anna}[1]{{}}   
 \newcommand{ \jyoti}[1]{\textcolor{lightbrown}{{\bf Jyoti: }#1}}

\newcommand{\remove}[1]{{}}
\newcommand{\changed}[1]{\textcolor{black}{#1}}
\newcommand{\finalchanged}[1]{{#1}}

\newcommand{\paragraphnew}[1]{\medskip\noindent{\bf#1}}

\title{The Complexity of Drawing a Graph in a Polygonal Region\thanks{This work was partially supported by NSERC and by ERC Consolidator Grant 615640-ForEFront.
A video explaining this work can be found at \texttt{https://youtu.be/JbmWLnY1hGk}.}}

\author{Anna Lubiw\inst{1}{\Envelope}  \and Tillmann Miltzow\inst{2} \and Debajyoti Mondal\inst{3}}
\authorrunning{Lubiw \and Miltzow \and Mondal} 

\institute{
Cheriton School of Computer Science, University of Waterloo, Waterloo, Canada  \email{alubiw@uwaterloo.ca}
\and 
Universit\'e libre de Bruxelles (ULB), Brussels, Belgium \email{t.miltzow@gmail.com}
\and 
Department of Computer Science, University of Saskatchewan, Saskatoon, Canada \email{dmondal@usask.ca}
}

\begin{document}

\maketitle
\begin{abstract}
We prove that the following problem is complete for the 
existential theory of the reals:
Given a planar graph and a polygonal region, with some vertices of the graph assigned to points on the boundary of the region, 
place the remaining vertices to create a planar straight-line drawing of the graph inside the region.
A special case is the problem of extending a partial planar graph drawing, which was proved  NP-hard by Patrignani.
Our result is \changed{one of} the first showing that 
a problem of drawing planar graphs with straight-line edges is hard for the existential theory of the reals.
The complexity of the problem is open 
for a simply connected region.
 
We also show that, even for integer input coordinates, it is possible that drawing a graph in a polygonal region requires some vertices to be placed at irrational coordinates. 
By contrast, the coordinates are known to be bounded in the 
special case of a convex region, or for drawing a path in 
any polygonal region.



\remove{
\anna{Here's the old version}
    In the study of geometric problems, the complexity class \ER 
    turned out to play a  crucial role. It exhibits a deep 
    connection between purely geometric problems and real algebra, 
    and is sometimes referred to as the real analogue to the class NP. 
    While NP can be considered as a class of computational problems that
    deals with existentially quantified \emph{boolean} variables, \ER 
    deals with existentially quantified \emph{real} variables.
    We study the problem of drawing a given planar graph inside
    a given polygonal region. We draw vertices as points and edges as
    line segments. Some of the vertices may be vertices of the 
    polygonal region. We denote the problem as \DGP.
    We show \DGP is \ER-complete.
     \jyoti{J:I wanted to focus more on the problem rather than the complexity class. Not sure whether it would be better to motivate this in terms of `compatible triangulations'. Here is an alternative:}

     \jyoti{new text follows:}    
    Partial drawing extension is a well-studied variation of the classic planar straight-line drawing problem that takes a planar graph $G$ along with a partial drawing  of the graph as input, and asks to decide whether the given partial drawing can be extended to a planar straight-line drawing of $G$. Although the problem is known to be  NP-hard, there are interesting scenarios when the problem is solvable in polynomial time. One classic example goes back to 1963, when Tutte showed that for every    3-connected graph $G$, any drawing of its outer face as a convex polygon can be extended to a straight-line drawing of $G$. A rich body of literature attempts to generalize this result, but many natural questions such as whether the original problem is $NP$-complete, or the problem is polynomial-time solvable for a given non-convex  drawing of the outer face, are all open for over a decade. 
    
     In this paper we prove that the `degenerate' version of partial drawing extension is \ER-complete, even in the fixed embedding setting. More specifically, we show that given a partial drawing of a planar graph $G$ as a polygonal region, it is \ER-complete to decide whether the drawing can be extended to a degenerate straight-line drawing of $G$, i.e., we allow a face to be drawn as a degenerate polygon. \jyoti{we may need to fix this based on how we edit the intro.}  This is a significant step towards  understanding the complexity of the partial drawing extension problem,  suggesting a deep connection with real algebra. 
      \jyoti{new text ends.}  \jyoti{can we show that the problem remains hard even with convex drawing constraint?}    
} 
\end{abstract}

\section{Introduction}
\label{sec:Introduction}

\remove{  Anna: for now, I'm, leaving everything in place, just "removed".  We can clean up later.
\changed{\DGP is defined as the problem to draw a given graph $G$ 
crossing-free inside a given closed polygonal region $R$. 
Some of the vertices of $G$ might be fixed on the boundary of $R$. 
Degeneracies, like a vertex on an edge are allowed.
We show \DGP is complete for the complexity class \ER. That is 
the problem is equivalent to the existential theory of the reals.
}

\till{I think the reader really wants to know a precise definition of the problem as early as possible. I think stating our main result on
page $4$ is too late. Therefore I would like to have the above as a first paragraph.}
\anna{You are just repeating the Abstract.  I do not think we should repeat the Abstract in the first paragraph of the paper.  Rather, the Introduction should introduce the problem.  The problem is defined (in words) on page 2, after the stage is set --  which I think is early enough.}
} 


There are many examples of structural results on graphs leading to beautiful and efficient geometric representations.  Two highlights are: Tutte's polynomial-time algorithm~\cite{tutte63} 
to draw any 
3-connected planar graph with convex faces inside any fixed convex drawing of its outer face; and Schnyder's tree realizer result~\cite{schnyder1990embedding} that provides a 
drawing of any $n$-vertex planar graph on an $n \times n$ grid.

On the other hand, there are geometric representations that are intractable, either in terms of the required coordinates or in terms of computation time. 
As an example of the former, a representation of a planar graph as touching disks 
(Koebe's theorem) is not always possible with rational numbers, nor even with roots of low-degree polynomials~\cite{Bannister2015galois}. 
As an example of the latter, Patrignani 
considered a generalization of Tutte's theorem and proved that it is NP-hard to decide 
whether a graph has a straight-line planar drawing 
when part of the drawing is fixed~\cite{patrignani2006extending}.
He was unable to show that the problem lies in NP because of coordinate issues.  
\remove{  Anna: postponing definition of \PDE until after \DGP
\changed{Patrignani and the graph drawing community 
denotes this problem as \PDE. In that problem it 
is forbidden to have degeneracies, like a vertex to be on an edge
or overlapping edges in the final drawing.}
\jyoti{I would rather not refer to the `community'.. How about - In graph drawing this problem is referred to as \DGP, where it is forbidden ...}
\anna{I prefer not defining \PDE yet.  We can make a small caps definition of it if you want that, but it should come later.
Part of my point is that we have a legitimate, very interesting problem.  We do not need to define it as a "degenerate" version of a previous problem.  So, I want to state our problem first, and relate it to \PDE afterwards.}
} 

This, and many other geometric problems, most naturally lie not in NP, but in a larger class, \ER, defined by 
\finalchanged{formulas in}
existentially quantified real (rather than Boolean) variables.
Showing that a geometric representation problem is complete for \ER is a stronger intractability result, often 
implying lower bounds on coordinate sizes.
For example, McDiarmid and M\"uller~\cite{mcdiarmid2013integer} showed that deciding if a graph can be represented as intersecting disks is \ER-complete.
 The relaxation from touching disks (Koebe's theorem) to intersecting disks  implies 
 that disk centers and radii can be restricted to integers, 
 but McDiarmid and M\"uller show that an exponential number of bits may be required.

In this paper we prove that an extension of Tutte's problem 
is \ER-complete. 
We call it the ``\DGP'' problem.  \changed{See Figure~\ref{fig:definition}.}
The input is a graph $G$ and a closed polygonal region $R$
(not necessarily simply connected), with some vertices of $G$ 
assigned fixed positions on the boundary of $R$.  
The question is whether $G$ has a straight-line planar drawing 
inside $R$ respecting the fixed vertices. 
We regard the region $R$ as a closed region which means that 
boundary points of $R$ may be used in the drawing.
\changed{A straight-line planar drawing (see Figure~\ref{fig:intro}(a,b)) means that 
vertices are represented as distinct points, and every edge is represented as a straight-line segment joining its endpoints, and no two of the closed line segments intersect except at a common vertex.
(In particular, no vertex point may lie inside an edge segment, and no two segments may cross.)}
\remove{ Anna: comparison will come later
\changed{In contrast to \PDE, we allow degeneracies in our drawings.}
\anna{No, we don't.  The graph drawing is a proper non-degernate graph drawing.  We only allow that edges touch the boundary. No graph drawing person is going to have any interest in allowing degeneracies, so we really should not write that.} \jyoti{Although I was thinking of degeneracy initially, but this is a good point. May be we point out in the conclusion that one possibility is to introduce a threshold such that edges must respect that distance from the  polygonal vertices that are not incident to that edge.}
\anna{Sure, the threshold version is interesting.}
} 

\begin{figure}[pt]
    \centering
    \includegraphics[width=.8\textwidth]
    {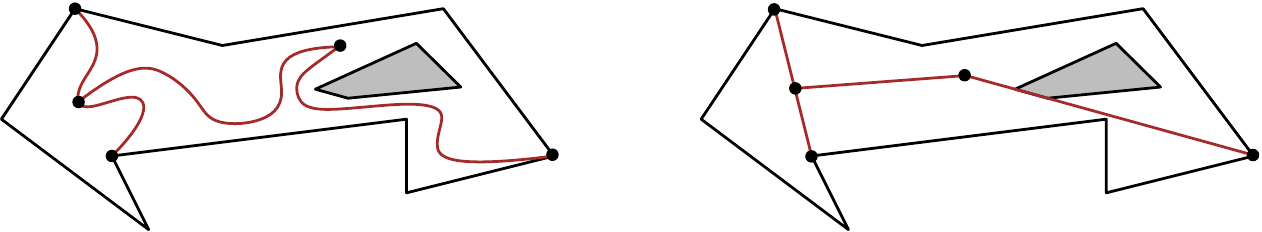}
    \caption{The \DGP problem.  Left: \changed{a polygonal region with one hole and a graph to be embedded inside the region.
    The three vertices on the boundary are fixed; the others are free.}
    Right: a straight-line embedding of the graph in the region.
    \changed{Note that we allow an edge of the drawing (in red) to include points of the region boundary.}
    \anna{I shrunk the figure a bit, but there's no text involved, so hopefully that's ok.  Otherwise, unshrink.}
    }
    \label{fig:definition}
\end{figure}

Furthermore, we give a simple 
instance of \DGP with 
integer coordinates where a vertex of $G$ may need 
irrational coordinates in any solution, \changed{thus defeating the naive approach to placing the problem in NP}.

\changed{The \DGP problem is a very natural one
that arises in 
practical applications such as 
dynamic and incremental graph drawing.  
Questions of the coordinates (or grid size) required for straight-line planar drawings of graphs are fundamental and well-studied~\cite{STL-handbook}.
It is surprising that a problem as simple and natural as \DGP is so hard and requires irrational coordinates.  
} 

\changed{We state our results below, but first we give some background on existential theory of the reals and on relevant graph drawing results. 
In particular, 
we explain that 
our problem is \changed{a generalization of}
the problem of extending a partial drawing of a planar graph to 
a straight-line drawing of the whole graph, called \PDE.  
See Figure~\ref{fig:intro}(c,d).}
\anna{If desired, we can define \PDE here or even explain the reduction here.}

\begin{figure}[pt]
    \centering
    \includegraphics[width=\textwidth]{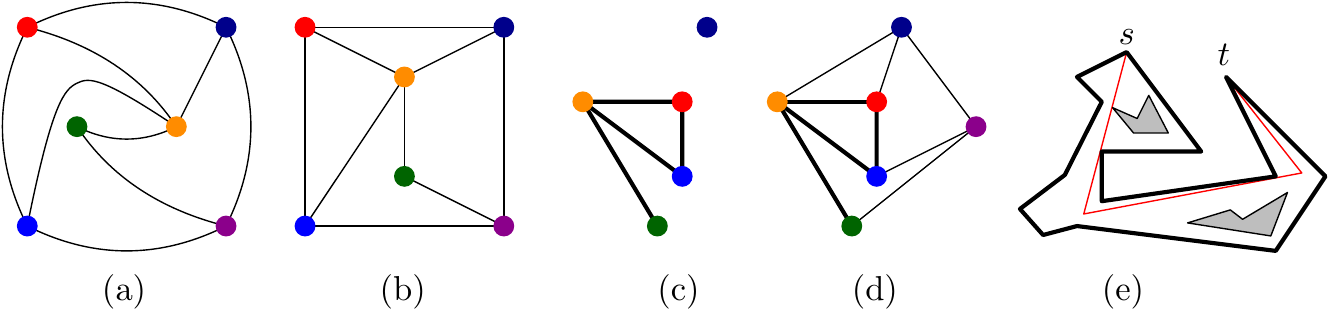}
    \caption{(a) A planar graph $G$. (b) A straight-line drawing of $G$. (c) A partial drawing $\Gamma$ of $G$. (d) Extension of $\Gamma$ to a straight-line drawing of $G$. (e) A minimum-link $s$-$t$ path in a polygonal region.
    %
    }
    \label{fig:intro}
\end{figure}

\paragraphnew{Existential Theory of the Reals.}
In the study of geometric problems, the complexity class \ER 
plays a  crucial role, connecting purely geometric problems and real algebra.
Whereas NP is defined in terms of 
\finalchanged{Boolean formulas in}
existentially quantified 
Boolean variables, \ER 
deals with \finalchanged{first-order formulas in}
existentially quantified real variables.  

Consider a 
first-order formula over the reals that contains only existential quantifiers,
 $\exists x_1,x_2,\ldots,x_n :$   
 $\Phi(x_1,x_2,\ldots,x_n)$, where $x_1,x_2,\ldots,x_n$ are real-valued variables and $\Phi$ is a quantifier-free formula involving equalities and inequalities of real polynomials. 
The \textsc{Existential Theory of the Reals} (\ETR) problem takes such a formula as an input and asks whether it is satisfiable. The complexity class \ER consists of all problems that reduce in polynomial time to \ETR. 
Many problems in combinatorial geometry and geometric graph representation naturally lie in this class, and furthermore, many have been shown  
to be \ER-complete, e.g., stretchability of a pseudoline arrangement~\cite{matousekSegments,Mnev,SchaeferS17},  recognition of segment intersection graphs~\cite{KratochvilM94} and disk intersection graphs~\cite{mcdiarmid2013integer}, computing the rectilinear crossing number of a graph~\cite{Bienstock91}, etc. 
For surveys on \ER, see~\cite{Schaefer09,cardinal2015computational,matousekSegments}.
A recent proof that the \AGP is \ER-complete~\cite{AbrahamsenAM18STOC} provides the framework we follow in our proof.

\paragraphnew{Planar Graph Drawing.}
The field of Graph Drawing investigates many ways of representing graphs geometrically~\cite{NishizekiRahman}, but we 
focus on 
the most basic representation of planar graphs, with points for vertices and straight-line segments for edges, such that 
segments intersect only at a common endpoint.
By F\'ary's theorem~\cite{fary}, every planar graph admits such a straight-line planar drawing.

In Tutte's famous paper,
 ``How to Draw a Graph,'' he gave a polynomial time algorithm to find a straight-line planar drawing of a graph
 \finalchanged{by first augmenting to a 3-connected graph}.
Given a combinatorial planar embedding 
(a specification of the faces) 
\finalchanged{of a 3-connected graph}
and given a convex polygon drawing of the outer face of the graph, his algorithm produces a planar straight-line drawing respecting both 
by reducing the problem to solving a linear system involving barycentric coordinates for each internal vertex. 
Tutte proved that the linear system has a unique solution and that the solution yields a drawing with convex faces. 
The linear system can be solved in polynomial time. 
For a discussion of coordinate bit complexity see Section~\ref{sec:bits}.

There is a rich literature on implications and variations of Tutte's result.  
We concentrate on the aspects of  
drawing a planar graph in a constrained region, or when part of the drawing is fixed. 
(We leave aside, for example, the issue of drawing graphs with convex faces, which also has an extensive literature.)

Our focus will be on straight-line planar graph drawings, but it is worth mentioning that 
without the restriction to straight-line drawings,
the problem of finding a planar drawing of a graph (with polygonal curves for edges) in a constrained region
is equivalent to the problem of extending a partial planar drawing (with polygonal curves for edges),
and there is a polynomial time algorithm for the decision version of the problem~\cite{AngeliniBFJKPR15}. 
Furthermore, there is an algorithm to construct such a drawing in which each edge is represented by a polygonal curve with linearly many segments~\cite{chan2015drawing}.

\changed{For the rest of this paper we assume straight-line planar drawings, which makes the problems harder.
The problem of drawing a graph in a constrained region is formalized as \DGP, defined above, and more precisely in Subsection~\ref{sec:contributions}.
The problem of finding a planar straight-line drawing of a graph after part of the drawing has been fixed
is called \PDE in the 
literature\finalchanged{---its complexity}
was 
formulated 
as an open question in~\cite{BrandenburgEGKLM03}.
}


The relationship between the two problems is that \DGP generalizes \PDE, as we now argue. 
Given an instance of partial drawing extensibility, for graph $G$ with fixed subgraph $H$, we construct an instance of $\textsc{Graph}$ $\textsc{in}$ 
$\textsc{Polygon}$ by making a point hole for each vertex of $H$ \finalchanged{and assigning the vertex to the point}.
Then an edge of $H$ can only be drawn as a line segment joining its endpoints, so we have effectively fixed $H$. 
To complete the bounded region $R$, we enclose the point holes in a large box.   
Clearly, we now have an instance of \DGP, and that instance has a solution if and only if $G$ has a planar straight-line drawing that extends the drawing of $H$. 
There is no easy reduction in the other direction because \DGP involves a closed polygonal region, so an edge may be drawn as a segment that  
touches, or lies on, the boundary of the region, and it is not clear how to model this as \PDE.
However, the version of \DGP~for an open region is equivalent to \PDE. 

We now summarize results on \PDE, 
beginning with positive results.
Besides Tutte's result that a convex drawing of the outer face can always be extended, there is a similar result for a star-shaped drawing of the outer face~\cite{HongN08}. 
There 
is also
a polynomial-time algorithm to decide the case when a convex drawing  of a subgraph 
is fixed~\cite{mchedlidze2016extending}.   Urhausen~\cite{Urhausen} examined the case when a  star-shaped drawing of one cycle in the graph is given, and proved that there always exists an extension with at most one bend per edge. Gortler et al.~\cite{gortler2006discrete} gave an algorithm, extending Tutte's algorithm, that succeeds in some (not well-characterized) cases  for a simple non-convex drawing of the outer face.   
The \PDE problem was shown to be NP-hard by Patrignani~\cite{patrignani2006extending}. 
\changed{This implies that \DGP is NP-hard.}
However, 
there are two natural questions about partial drawing extensibility that remain open: (a) does the problem belong  to the class NP \changed{(discussed in detail by Patrignani~\cite{patrignani2006extending})}, and (b) does the problem remain NP-hard when a combinatorial embedding of the graph is given and must be respected in the drawing.
Our results shed light on these questions for the 
\changed{more general \DGP problem:}
the problem cannot be shown to lie in NP by means of giving the vertex coordinates, and the problem is still \ER-hard when a combinatorial embedding of the graph is given.




\remove{
\paragraph{Partial Drawing extensibility:} Tutte's result on 3-connected graph drawing~\cite{tutte63} can be viewed as one of the initial results on the \emph{partial drawing extensibility} problem that given a partial straight-line drawing of a planar graph, seek for an extension of the drawing to a complete drawing of the input graph. In addition to its theoretical appeal, the problem finds application in the context of dynamic and interactive graph drawing~\cite{mchedlidze2016extending}, 
 mesh generation~\cite{BrandenburgEGKLM03}, and computing compatible triangulations of  polygons~\cite{AronovSS93,LubiwM17}.  \jyoti{J-to-Anna, I am not sure whether it would help to move the focus to compatible triangulations, or adding more details of how we reached to G-in-poly.  I will leave that to you, feel free to rewrite as it suits best. There are also many results on simultaneous drawing extensions and other representation extensions. Should we mention it in 'further backgrounds'?}
 \anna{I think we should leave out discussion of compatible triangulations and also of simultaneous drawings.  It's too much!}

In 2004, the question of whether the partial drawing extensibility is polynomial-time solvable or not was posted in a collection of  interesting open problems in graph drawing~\cite{BrandenburgEGKLM03}. Patrignani~\cite{patrignani2006extending} showed that the problem is NP-hard\footnotetext{The conference version of the paper claimed NP-completeness}. However, there are two natural questions yet to resolve: (a) whether the problem belongs to the class NP, and (b) whether the problem remains NP-hard in the fixed embedding setting, i.e., when the input consists of a planar graph along with a combinatorial embedding, and the partial and extended drawings both must respect the given embedding. In the topological setting, i.e., when the edges are not required to be straight line segments, the partial drawing extensibility is  solvable in polynomial time~\cite{AngeliniBFJKPR15} both in the variable and fixed embedding settings.
}

Besides Tutte's result, there is another special case 
of \DGP
that is well-solved, namely when the graph is just a path with its two endpoints $s$ and $t$ fixed on the boundary of the region.
\changed{See Figure~\ref{fig:intro}(e).}
This problem is equivalent to the {\sc Minimum Link Path} problem\finalchanged{---to find a path from $s$ to $t$ inside the region with a minimum number of segments.  This is because}  a path of $k$ edges can be drawn inside the region if and only if the minimum link distance between $s$ and $t$ is less than or equal to $k$.  Minimum link paths in a polygonal region can be found in polynomial time~\cite{mitchell1992minimum}, and in linear time for a simple polygon~\cite{suri86}.
\changed{The complexity of the coordinates is well-understood (see Section~\ref{sec:bits}).} 



\remove{
\paragraph{Minimum-Link Paths:}  The partial drawing extensibility problem shares some striking similarities with the \emph{minimum-link path} problem. The input of the  minimum-link path problem is a \emph{polygonal region} $P$ (i.e., a polygon with holes) along with a pair of points $s,t$ inside $P$, and the goal is to compute a polygonal path with minimum number of links. The decision version of this problem seems to be a special case of the  partial drawing extensibility problem, where the polygonal domain plays the role of a partial drawing, and the vertices and edges that are not yet drawn is a path between $s$ and $t$. 

Kostitsyna~\cite{kostitsyna2016complexity} showed that the minimum-link path problem is NP-hard by reducing the 2-partition problem. 
\anna{The min-link path problem can be solved in poly. time.  It only gets hard in higher dimensions, or \sout{with reflections allowed} or with bends restricted to edges.  Please double-check.} \jyoti{J: we are talking about polygonal regions (not simple polygons), that is why it is hard.}
\anna{What Kostitsyna et al. prove NP-hard is MinLinkPath(1,2), where the bends are restricted to edges.  Min link path in a polygonal domain is what they call  MinLinkPath(2,2) and is solved in nearly quadratic time in their ref. [41] by Mitchell, Rote, Woeginger, 1992, ``Minimum-link paths among obstacles in the plane''.}
However, this does not subsumes Patrignani's~\cite{patrignani2006extending} result since in the minimum-link path model, the vertices are allowed to lie on the edges, which forbidden in the original partial drawing extensibility model.  Kostitsyna's~\cite{kostitsyna2016complexity} reduction uses `0-width slits', which corresponds to vertex overlaps, and like Patrignani's~\cite{patrignani2006extending} result, the hardness result does not apply to the fixed embedding setting. When restricted to simple polygons (i.e., without holes), the minimum-link path problems can be solved in polynomial time~\cite{suri86}.
}

\subsection{Our Contributions}
\label{sec:contributions}
Our problem is defined as follows.
\FrameSep5pt 
\begin{framed}
\noindent {\bf \DGP}  

\noindent {\bf Input:} A planar graph $G$ and a polygonal region $R$ 
with some vertices of $G$ assigned to fixed positions on the boundary of $R$.

\noindent {\bf Question:} Does $G$ admit a planar straight-line drawing inside $R$ respecting the fixed vertices? 
\end{framed}

The graph may be given abstractly, or via a \emph{combinatorial embedding}  
which specifies the cyclic order of edges around each vertex, thus determining the faces of the embedding.
When a combinatorial embedding is specified then the final drawing must respect that embedding.

Note that we regard $R$ as a closed region.  Thus, points on the boundary of $R$ may be used in the drawing of $G$.  In particular, an edge of $G$ may be drawn as a segment that touches, or lies on, the boundary of $R$.  See Figure~\ref{fig:definition}.  
Note that we still require the 
drawing of $G$ to be ``simple'' in the conventional sense that no two edge segments may intersect except at a common endpoint.  

Our first result is that
solutions to \DGP may involve irrational points.  This will in fact follow from the proof of our main hardness result, but it is worth seeing a simple example.

\begin{restatable}{theorem}{IrrationalVertices}
\label{thm:IrrationalVertices}
 There is an instance of \DGP 
 with all coordinates given by integers,
 in which some vertices need irrational coordinates.
\end{restatable}


Note that the theorem does not rule out membership of the problem in NP, since it may be possible to demonstrate that a graph can be drawn in a region without giving explicit vertex coordinates.  
We prove Theorem~\ref{thm:IrrationalVertices} by adapting an example from Abrahamsen, Adamaszek and Miltzow~\cite{abrahamsen2017irrational}
that proves a similar irrationality result for the \finalchanged{$\AGP$}. 
Further discussion of bit complexity for special cases of the problem can be found in Section~\ref{sec:bits}.

Our main result is
the following, which holds whether the graph is given abstractly or via a combinatorial embedding.

\begin{restatable}{theorem}{MainER}\label{thm:Main}
 \DGP is \ER-complete.
\end{restatable}

\remove{  Anna: I'm removing this for now.  Let's save it for the final version of the paper.
\till{Do we want to make statements about universality?} 
\till{With Universality i mean: For every semi-algebraic set there exists an instance of \DGP, such that there solution space is equivalent in some specified topological  and algebraic sense. As a corollary, we can enforce \emph{every} algebraic number as the $x$ coordinate of a valid solution.
Would that be something interesting to state?} 
\anna{If universality is true and follows readily from our proof, it would be good to say so.  However, maybe it isn't the first priority (it will not change whether our paper is accepted or not).  Also, I thought you could not get universality for the Art Gallery?  Maybe I'm wrong about that?  Or maybe this problem is different in that respect?}
\till{Anna, you are right, we could not get universality in the usual sense, as guards are unlabelled. Vertices however are labeled, so we get all the universality we want very easily, for our paper. :)}
}

We prove Theorem~\ref{thm:Main} using a reduction from 
a problem called \etrinv which was introduced and proved \ER-complete by Abrahamsen, Adamaszek and Miltzow~\cite{AbrahamsenAM18STOC}.

\begin{definition}[\etrinv]
\label{def:etrinv}
In the problem $\etrinv$, we are given a set of real 
variables $\{x_1,\ldots,x_n\}$, and a set of equations of the form
$
x=1,\quad x+y=z,\quad x\cdot y=1,
$
for $x,y,z \in \{x_1, \ldots, x_n\}$.
The goal is to decide whether the system of equations has 
a solution when each variable is restricted to the range $[1/2,2]$.
\end{definition}

Reducing from \etrinv, rather than from ETR, has several crucial advantages.  
First, we can assume that all variables are in the range $[1/2,2]$.
Second, we do not have to implement a gadget that simulates 
multiplication, but only inversion, i.e., $x\cdot y = 1$.
For our purpose of reducing to \DGP, we will find it useful to
further modify \etrinv to avoid equality and to ensure planarity of the variable-constraint incidence graph, as follows:

\remove{
Dobbins, Kleist, Miltzow and Rz{\polhk{a}}{\.{z}}ewski~\cite{AreaUni}
\changed{strengthened the result to show}
\ER-hardness for the case that the variable-constraint incidence graph 
is planar. We use here a similar, but
slightly different idea.
We will later explain the subtle difference.
\jyoti{This is confusing in the sense that we have an addition gadget, as in Fig 12. Why we could not use their transformation and replace that with our addition gadget? In the addition gadget section we mentioned something like `single' gadget. Is it hard to compare without knowing Dobbin et al.'s result in detail? May be we need a better explanation to motivate the new \planetrinv proof.}
\till{I cannot explain the difference before we defined
out planaretrinv.}
} 


\begin{definition}[\planetrinv]
\label{def:planaretrinv} In the problem \planetrinv, we are given a set of real variables
 $\{x_1,\ldots, x_n\}$, and a set of equations and inequalities of the form
$
x = 1\text{, } x + y \le z\text{, } x + y \ge z\text{, } x \cdot y \le 1\text{, } x \cdot y \ge 1\text{, }
\text{ for } x, y, z \in \{x_1,\ldots, x_n\}.
$
 Furthermore, we require planarity of the \emph{variable-constraint incidence graph}, which is the bipartite graph that has a vertex for every variable and every constraint and an edge when a variable appears in a constraint. 
 The goal is to decide whether the system of equations has a solution when each variable is restricted to lie in $[1/2,4]$. 
\end{definition}


As a technical contribution, we prove the following. 
\begin{restatable}{theorem}{PlanarETRINV}\label{thm:planarETRINV}
 \planetrinv is \ER-complete.
\end{restatable}

The proof, which is in Appendix~\ref{sec:PlanarETR}, 
builds on the work of
Dobbins, Kleist, Miltzow and Rz{\polhk{a}}{\.{z}}ewski~\cite{AreaUni}
who showed
that \etrinv is \ER-complete even 
\changed{when} the variable-constraint incidence graph is planar.
\changed{We cannot use their result directly, but follow similar steps in our proof.}
\remove{We denote this as \planetrinvOld.
One might be tempted to believe that the following 
reduction from the \planetrinvOld shows 
\planetrinv is \ER-complete. Take an instance $I$ of
\planetrinvOld and replace every addition constraint 
$x+y =z$ by two constraints of the form
$x+y \leq z$ and $x+y \geq z$.
And the same for the inversion gadget. 
The problem with this idea is that the new variable-constraint graph is \emph{not} planar anymore.
\till{I hope things are clear now.}
}

\section{Irrational Coordinates}

\IrrationalVertices*

In fact, the result follows 
from our proof of Theorem~\ref{thm:Main}, but it is interesting to have a simple explicit example,
which is given in Figure~\ref{fig:irrationalVertices}.  
This example is adapted from a result of Abrahamsen et al.~\cite{abrahamsen2017irrational}.  
Details can be found 
in Appendix~\ref{app:irrational},
but we outline the idea here.
%
Abrahamsen et al. studied the \AGP, where given a polygon $P$ and a number $k$, and we want to find a set of at most $k$ guards (points) that together see the entire polygon. We say a guard $g$ sees a point $p$ if the entire line-segment $gp$ is contained inside the polygon $P$. Abrahamsen et al. gave a simple polygon with integer coordinates such that there exists only one way to guard it optimally, with three guards. Those guards have irrational coordinates. See Figure~\ref{fig:irrationalGuards} 
for a sketch of their polygon. A key ingredient of their construction is to create notches in the polygon boundary that force there to be a guard on each of the three so-called \emph{guard segments}. 
The coordinates of the polygon then force the guards to be at irrational points.

\begin{figure}[t]
    \centering
    \includegraphics[width=.5\linewidth]{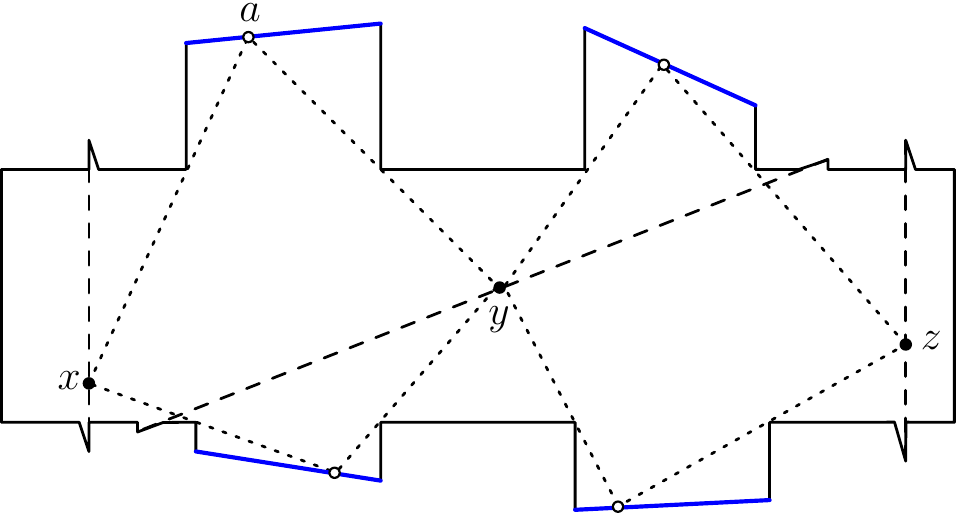}
    \caption{A sketch of the polygon from Abrahamsen et al. 
    The three guards (black dots) must lie on the \emph{guard segments} (dashed lines).
    }
    \label{fig:irrationalGuards}
\end{figure}

\begin{figure}[t]
    \centering
    \includegraphics[width=.8\textwidth]{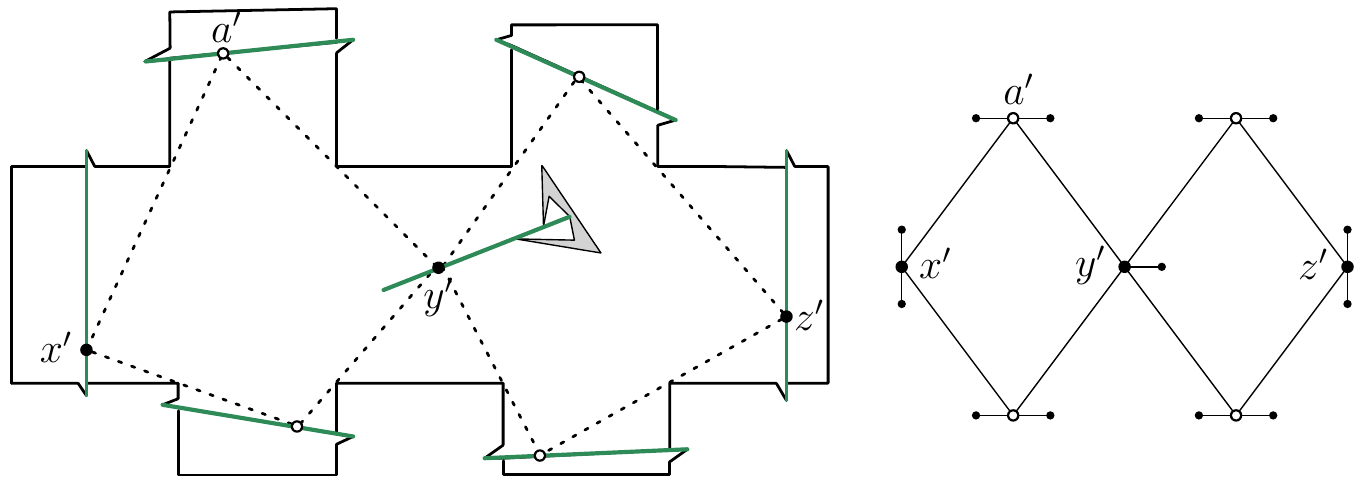}
    \caption{Left: An instance of \DGP based on Figure~\ref{fig:irrationalGuards} that requires vertices at irrational coordinates. 
    Right: The graph, with small dots indicating the fixed vertices.
    }
    \label{fig:irrationalVertices}
\end{figure}

We adapt their example 
by using \emph{variable segments} (shown in green) instead of guard segments, and vertices instead of guards.
By placing notches in the polygon boundary with fixed vertices of the graph in the notches, we can force there to be a vertex on each variable segment.
We create two cycles that replicate the guarding constraints, and  use a hole in order to keep our graph drawing planar.
From their proof we show that $x'$, $y'$ and $z'$ must be at irrational coordinates.


\section{\ER-completeness}\label{sec:MainER}
\MainER*

\begin{proof}
First note that \DGP lies in \ER since we can express it as \finalchanged{an \ETR{} formula}.
To prove that the problem is \ER-hard we give
a reduction from \planetrinv. 
Let $I$ be an instance of \planetrinv.
We will build an instance $J$ of \DGP such that $J$ admits an affirmative answer if and only if $I$  is satisfiable. 
The idea is to construct gadgets to
represent variables, and gadgets to enforce the addition and inversion inequalities, 
$x + y \leq z ,  x + y \geq z, x\cdot y \leq 1,  x\cdot y \geq 1$.
We also need gadgets to copy and 
\finalchanged{replicate}
variables\finalchanged{---``wires'' and ``splitters'' as conventionally used in reductions.}
Thereafter, we have to describe how to combine 
those gadgets to obtain 
an instance $J$ of \planetrinv.

\remove{
\begin{figure}[htbp]
    \centering
    \includegraphics[width=\textwidth]{DefineVariables}
    \caption{\anna{I propose deleting this figure and using the one below instead.} The value of a variable restricted to be in $[1/2,4]$  can be represented by 
    the position of a vertex. The positions marked with $\frac12$ 
    and $4$ represent those respective values. Every other 
    value in the interval $[\frac12,4]$ is represented by linear
    interpolation.
    }
    \label{fig:variables}
\end{figure}
}

\paragraphnew{\finalchanged{Encoding Variables.}}
We will encode the value of a variable in $[1/2,4]$ as the position of a vertex that is constrained to lie on a line segment of length $3.5$,   
which we call a \emph{variable-segment}.
One end of a variable-segment encodes the value $\frac{1}{2}$, the other end encodes the value $4$, and linear interpolation fills in the values between.
Figure~\ref{fig:segment-end} shows one side of the construction that 
forces a vertex to lie on a variable-segment.  The other side is similar.

\begin{figure}[pb]
    \centering
    \includegraphics[width=.4\linewidth]{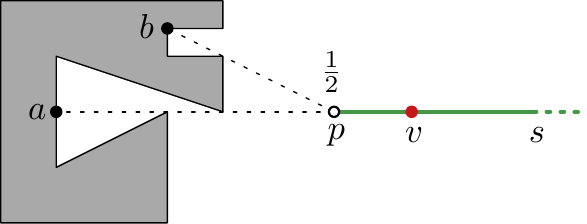}
    \caption{Variable $v$ is represented as a point on variable-segment $s$ (shown in green). The construction of one end of $s$ is illustrated.
In the graph, vertex $v$ is adjacent to fixed vertices $a$ and $b$ on the boundary of a hole of the region (shaded).  Adjacency with $a$ forces $v$ to lie on the line of $s$.  Adjacency with $b$ forces $v$ to lie at, or to the right of, point $p$ which is associated with the value $1/2$. \finalchanged{Note that $p$ is not a vertex.}  
}
    \label{fig:segment-end}
\end{figure}

\remove{  some long discussion . . . 
\anna{We really must describe how the region boundary behaves to create each segment.  I guess we can place the segments as we like in the plane and then place tiny holes of the region at the endpoints of the segment like the hole in Figure~\ref{fig:irrationalVertices}.  I think we need a hole at each end.  Will we allow the vertex to be placed coincident with the boundary vertex that is creating the segment?  If not, then we need some other way of enforcing the endpoints of the segment. I can imagine doing this by increasing the degree of the vertex, but that seems a bit dangerous 
because the vertex will need other edges incident to it.}
\till{I understand your concern. Bot solutions will actually work. 
Solution 1: Don't do anything. The vertex might lie outside of its suggested range.
So what. If there is ANY solution at all THEN there is also a solution inside the
correct range :).
\anna{Well, where do we argue that for each gadget?  In general, I think our proof is not rigourous enough.}
Solution 2: add some extra stuff to actually enforce the variable vertex
to lie at the correct position. This is fine, but require more attention 
to check. 
Personally, I prefer solution 1, as I don't like to add stuff 
that is not needed. Minimalism yeahh ;)
\anna{I suspect that solution 2 is more minimal.  We only need to describe in this section how to restrict a variable to its proper interval.  Whereas with solution 1 we need to look at each gadget and make the argument that if there is an out-of-range solution then there is also an in-range solution.  This seems like more work. I also imagine that solution 1 will make reviewers wonder (just as I did), but solution 2 will be more soothing
since we are doing exactly what \planetrinv says.}
About having a connection to the top and to the bottom
is probably more save. We need to prevent sometimes that an edge is above rather
than below a hole, so we need to bound its range somewhat. }
}

By slight abuse of notation, we will identify 
a variable and the vertex representing it, if there is no ambiguity.
For the description of the remaining gadgets, our figures will show variable-segments (in green) without showing the polygonal holes that determine them.
%

\begin{figure}[ht]
    \centering
    \includegraphics[width=.8\textwidth]{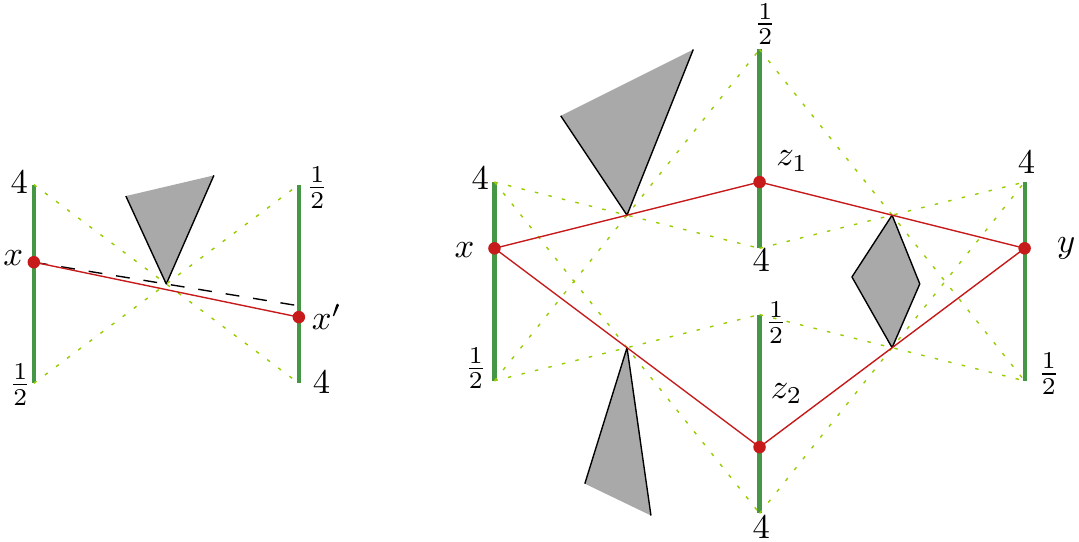}
    \caption{Copying.
    Left: a gadget to enforce $x {\leq} x'$.  Right: The full gadget enforcing $x{=}y$.
    }
    \label{fig:FullCopy}
\end{figure}

\paragraphnew{\finalchanged{Copy gadget.}}
Given a variable-segment for a variable $x$, we will need to transmit its value \finalchanged{along a ``wire''}  to other locations in the plane. 
\finalchanged{We do this using a copy gadget in which we construct}
a 
variable-segment for a new variable $y$ and enforce $x=y$.
We show how to construct a gadget that ensures $x \leq x'$ for a new variable $x'$, and then combine four such gadgets, enforcing
$x  \leq  z_1, z_1  \leq  y , x  \geq  z_2 , z_2  \geq  y$.
This implies $x=y$.

The gadget enforcing $x \leq x'$ 
is shown at the left of 
Figure~\ref{fig:FullCopy}. It consists of 
two parallel 
variable-segments.
In general, these two segments need not be horizontally aligned.
In the graph we connect the corresponding vertices by an edge.
The left and the right variables are
encoded in opposite ways, i.e., $x$ increases as the vertex moves up
and $x'$ increases as the corresponding vertex goes down.
We place a hole of the polygonal region (shaded in the figure) with a vertex at the intersection point of the lines joining the top of one variable-segment to the bottom of the other. The hole must be large enough that the edge from $x$ to $x'$ can only be drawn to one side of the hole.
An argument about similar triangles, or the ``intercept theorem'', 
also known as Thales' theorem, implies $x \leq x'$. 

We combine four of these gadgets to construct our copy gadget,  as illustrated on the right of Figure~\ref{fig:FullCopy}.
%

\begin{figure}[pb]
    \centering
    \includegraphics[width=.7\textwidth]{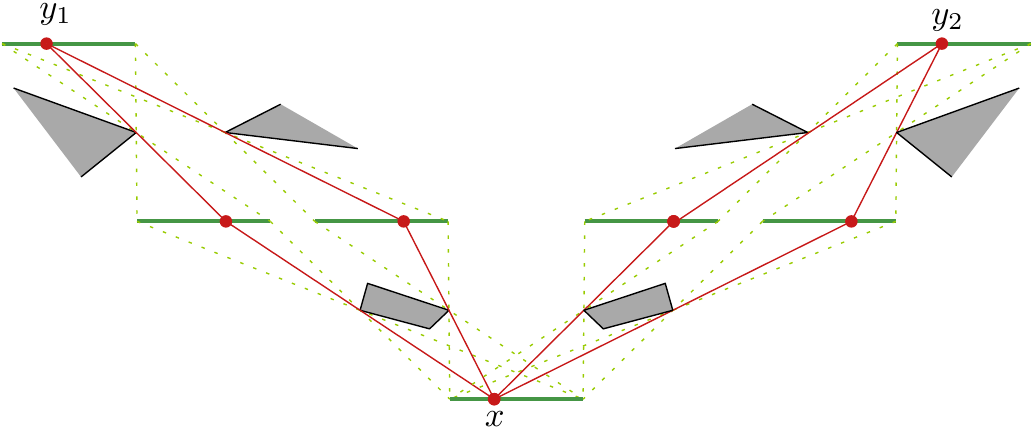}
    \caption{
    Splitting.
    The variables $y_1,y_2$ have both the same value as $x$.}
    \label{fig:Duplicate}
\end{figure}

\paragraphnew{\finalchanged{Splitter gadget.}}
Since a single variable may appear in several constraints,
we 
\finalchanged{may need to}
split a wire into two wires, each holding the correct value of the same variable. 
Figure~\ref{fig:Duplicate} shows a gadget to 
\finalchanged{split} 
the variable $x$ to variables $y_1$ and $y_2$.  The gadget consists of two copy gadgets sharing the variable-segment for $x$.
We can construct the two copy gadgets to avoid any intersections between them.

\paragraphnew{\finalchanged{Turn gadget.}}
We need to encode a variable both as a vertical and
as a horizontal variable-segment. 
To transform one into the other we use a turn gadget.

\begin{figure}[pt]
    \centering
    \includegraphics[width=.9\textwidth]{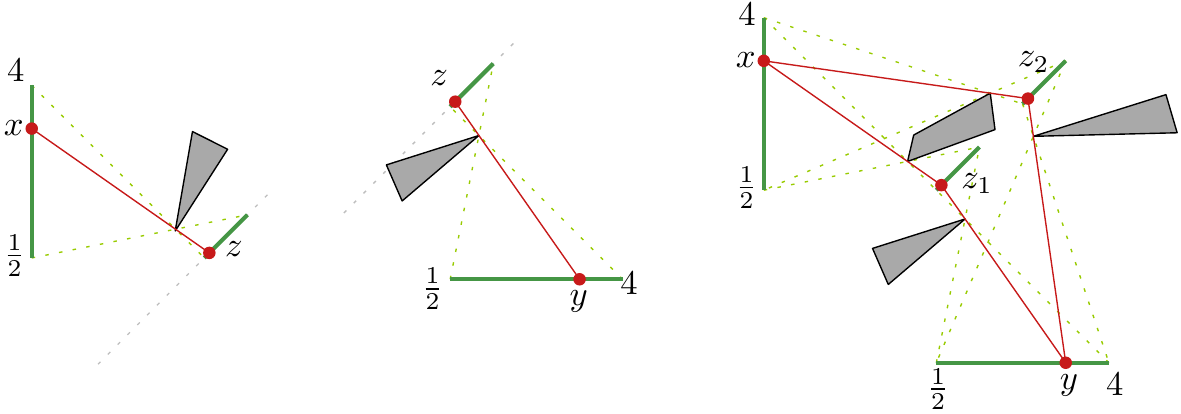}
    \caption{
Turning. Left: Gadget to encode $x\leq f(z)$. Middle: Symmetric gadget to encode $y\geq f(z)$. Right: Four gadgets of the 
    previous type combined to enforce $x=y$, for $x$ and $y$ on a vertical and horizontal variable-segment, respectively.
    }
    \label{fig:turn}
\end{figure}

The key idea is to construct 
two diagonal variable-segments for variables $z_1$ and $z_2$, and then transfer the value of the vertical variable-segment to the horizontal variable-segment using $z_1,z_2$. This is in fact very similar to the 
copy gadget, except that the intermediate variable-segments are placed \changed{on a line of slope 1.}  
We do not know if it is possible to enforce the 
constraint $x\leq z$ directly. 
However, it is sufficient to enforce \changed{$x \leq f(z)$} for some  function $f$. 
\remove{\anna{Can't we write this as $x \leq f(z)$?  Then we could avoid the last step below (about strictly monotone and bijective).}
\till{Interesting point. So, I thought about it and 
at first I thought that we really need monotonicity,
as the statement is otherwise not true. But then your argument works and I could not see any flaw in it.
So my conclusion is that, we use monotonicity also
in your argument implicitly. When we write $x \leq  f(z)$,
we say all values $x$ smaller then $f(z)$ are fine. }}
See 
the left side of Figure~\ref{fig:turn}. 
Interestingly, we don't even know the function $f$. However, we do know that $f$ is monotone and we can construct 
another gadget enforcing \changed{$y \geq f(z)$}, 
for the same function $f$,
by making another copy of the first gadget reflected through the line of the variable-segment for $z$.

Combining four such gadgets, as on the right of Figure~\ref{fig:turn}, 
 yields the following inequalities:
 $x \leq f_1(z_1), f_1(z_1) \leq y, y \leq f_2(z_2), f_2(z_2) \leq x$.
This implies $x = y$.

\paragraphnew{Addition \finalchanged{gadget}.}
%
\begin{figure}[btp]
    \centering
    \includegraphics[width=.8\textwidth]
    {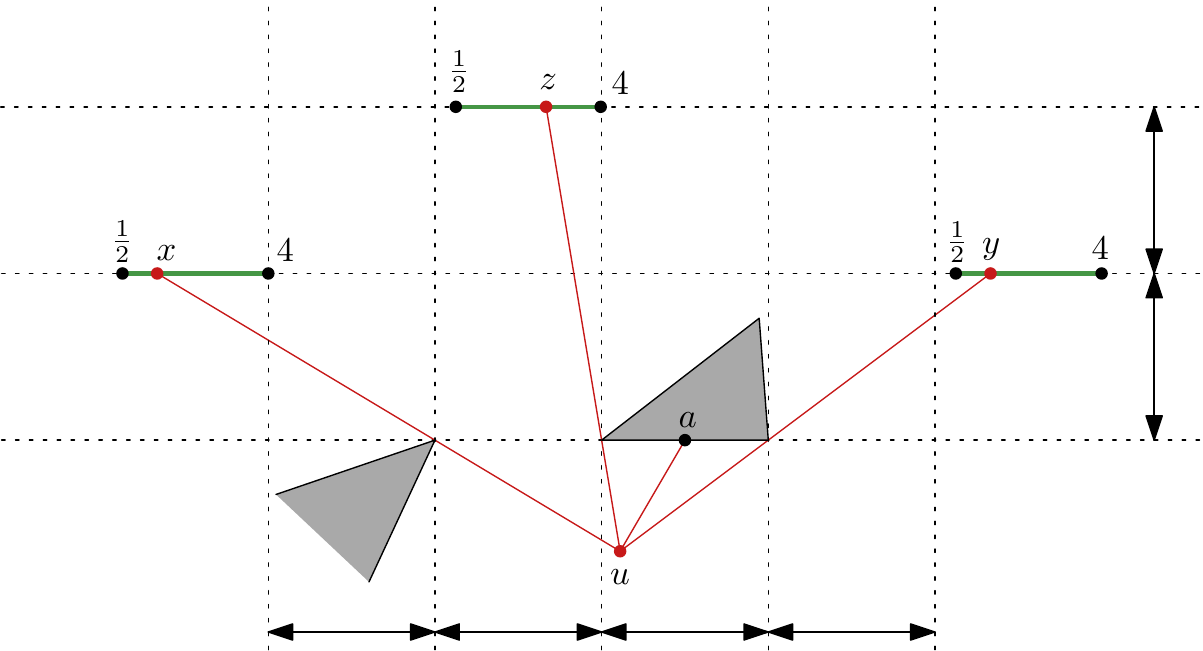}
    \caption{
    Addition. The three vertices $x,y,z$ can only be connected 
    to $u$ if  $x {+} y {\geq} z$ holds.}
    \label{fig:Addition}
\end{figure}
The gadget to \finalchanged{enforce} 
$x + y \geq z$ is depicted in
Figure~\ref{fig:Addition}. 
Important for correctness is that the gaps between the dotted auxiliary
lines have equal lengths. 
This is essentially the same
gadget that was used 
by Abrahamsen et al.~\cite[Lemma 31]{AbrahamsenAM18STOC}.
We offer an alternative correctness proof, which is in   Appendix~\ref{sec:addition}.
\begin{lemma}[\cite{AbrahamsenAM18STOC}]\label{lem:addition}
    The gadget in Figure~\ref{fig:Addition} enforces 
    $x+y\geq z$.
\end{lemma}

\noindent
The gadget that enforces $x+y \leq z$ 
is just a mirror copy of the previous gadget.

\paragraphnew{Inversion \finalchanged{gadget}.}
The inversion gadgets to enforce
$x \cdot y \leq 1$ and $x \cdot y \geq 1$
are 
depicted in Figure~\ref{fig:inversion}.
We use a horizontal variable-segment for $x$ and a vertical variable-segment for $y$ and align them as shown in the figure, $1.5$ units apart both horizontally and vertically.  We make a triangular hole with its apex at point \finalchanged{$q$} as shown in the figure.
The graph has an edge between $x$ and $y$.
\begin{figure}[tb]
    \centering
    \includegraphics[width=.9\textwidth]{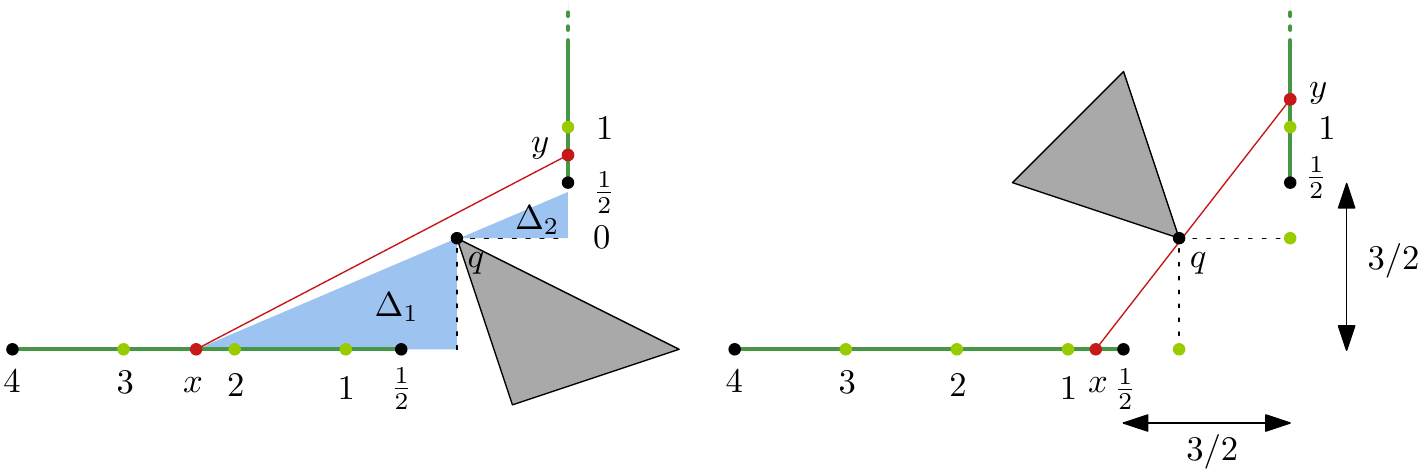}
    \caption{
    Inversion.  
    Left: Gadget enforcing $x\cdot y \geq 1$. Right: Gadget enforcing $x\cdot y \leq 1$.
}
    \label{fig:inversion}
\end{figure}

For correctness, observe that if $x$ and $y$ are positioned so that the line segment joining them goes through point \finalchanged{$q$}, then, because triangles $\Delta_1$ and $\Delta_2$ (as shown in the figure) are similar, we have $\frac{x}{1} = \frac{1}{y}$, i.e.~$x \cdot y=1$.
If the line segment goes above point \finalchanged{$q$} 
(as in the left hand side of Figure~\ref{fig:inversion})
then $x\cdot y \geq 1$, and if the line segment goes below then $x \cdot y \leq 1$.

\paragraphnew{Putting it all together.}
It remains to show how to obtain 
an instance of $\textsc{Graph}$ $\textsc{in}$ 
$\textsc{Polygon}$  in polynomial time from an instance of \planetrinv.

Let $I$ be an instance of \planetrinv. 
As a first step we \finalchanged{modify the planar variable-constraint incidence graph so that a variable vertex of degree $d$ is replaced by ``splitter'' vertices of degree at most 3 to create $d$ copies.}
Then we compute a plane rectilinear drawing $D$ of the  \finalchanged{resulting planar} 
graph, which can be done in polynomial time using rectilinear planar drawing algorithms~\cite{NishizekiRahman}. 
The edges of $D$ act as wires and we replace each horizontal and vertical segment by a copy gadget, and   
replace every $90$ degree corner, by a turn gadget. 
\finalchanged{Every splitter vertex and constraint vertex}
will be replaced by the corresponding gadget, \finalchanged{possibly using turn gadgets.}
We add a big square to the outside, to ensure that everything is inside one polygon.  
 \changed{See Figure~\ref{fig:reduction} in Appendix~\ref{sec:addition}.}

It is easy to see that this can be done in polynomial time, 
as every gadget has a constant size description and can be 
described with rational numbers, although, we did not do it explicitly. 
In order to see that we can also use integers, note that we can scale everything with the least common multiple of all the denominators of all numbers appearing. This can also be done in polynomial time.
\qed
\end{proof}

\section{Vertex Coordinates}\label{sec:bits}

Since we have shown that \DGP may require irrational coordinates for vertices in general, it is interesting to examine bounds on coordinates for special cases.
In this section we discuss the bit complexity of vertex coordinates needed for two well-solved special cases of \DGP.

Tutte's algorithm~\cite{tutte60} finds a straight-line planar drawing of a graph inside a fixed convex drawing of its outer face.  
Suppose the graph has $n$ vertices and each coordinate of the convex polygon uses $t$ bits.
Tutte's algorithm runs in polynomial time, but the number of bits used to express the vertex coordinates is a polynomial function of $t$ and  $n$.  
The dependence on $n$ means that the drawing uses ``exponential area.''
Chambers et al.~\cite{ChambersEGL12} gave a different algorithm that uses polynomial area---the number of bits for the vertex coordinates is bounded by a polynomial in $t$ and $\log n$.

\remove{  
(As an aside, the algorithm of Chambers et al.~loses strict convexity of the faces for 3-connected graphs.  B\'ar\'any and Rote~\cite{barany2006strictly} achieved strictly convex faces and polynomial area, but did not allow a fixed outer polygon. 
It is an open problem to achieve all three properties, i.e.~to find a straight-line planar drawing of a 3-connected graph with a fixed outer face, convex inner faces, and polynomial area. The three cited papers each achieve two of the properties.)
}


The other well-solved case of \DGP is the minimum link path problem. 
Here we have a general polygonal region with holes, $R$, but the graph is restricted to a path with endpoints $s$ and $t$ fixed on the boundary of $R$.
Based on a lower bound result of Kahan and Snoeyink~\cite{kahan1999bit}, Kostitsyna et al.~\cite{kostitsyna2016complexity} proved a tight bound of $\Theta(n\log n)$ bits for the coordinates of the bends on a minimum link path.  
Note that the dependence on $n$ means that this bound is exponentially larger than the bound 
for drawing a graph inside a convex polygon.
Problem 3 below asks about the complexity of drawing a tree in a polygonal region.

\section{Conclusion and Open Questions}\label{sec:Conclusion}

Our result that \DGP is \ER-complete is one of the first \ER-hardness results about drawing planar graphs with straight-line edges---along with a recent result 
about drawings with prescribed face areas~\cite{AreaUni}.
We conclude with some open questions:

\remove{
In this paper, we studied the problem of 
finding a planar straight-line drawing of a graph inside a polygonal region, and showed that it is \ER-complete.
\changed{Previous \ER-hardness results for graph drawing involved other representations, such as disk or segment intersection, 
or involved straight-line drawings of non-planar graphs.
Our result is one of the first \ER-hardness results about drawing planar graphs with straight-line edges---along with a recent result 
about drawings with prescribed face areas~\cite{AreaUni}.}
%


We conclude with some open questions:
}


\noindent{\bf 1.} Our proofs of Theorems~\ref{thm:IrrationalVertices} and~\ref{thm:Main} used the fact that the polygonal region may have holes and may have collinear vertices.  Is \DGP polynomial-time solvable for a simple polygon (a polygonal region without holes) whose vertices lie in general position (without collinearities)?
    
\noindent{\bf 2.} Our proofs also used the assumption that the polygonal region is closed. For an open 
    region, the problem \DGP is equivalent to the problem \PDE.  Is this problem \ER-hard? There are two versions, depending on whether the graph is given abstractly or via a combinatorial embedding. In the first case the problem is known to be NP-hard~\cite{patrignani2006extending}, but in the second case even that is not known.
    
\noindent{\bf 3.} What is the complexity of \DGP when the graph is a tree?  Can vertex coordinates still be bounded as for the minimum link path problem?
    When the tree is a caterpillar, the problem might be related to the minimum link watchman tour problem, which is known to be NP-hard~\cite{DBLP:journals/ipl/ArkinMP03}.  


\subsubsection{Acknowledgment.}
We would like to thank G\"{u}nter Rote, who discussed
with the \changed{second} author the turn gadget in the context of the \AGP.

\bibliographystyle{splncs04}
\bibliography{Bibli}

\newpage

\appendix

\section{\planetrinv is \ER-complete}\label{sec:PlanarETR}
The purpose of this section is to prove Theorem~\ref{thm:planarETRINV}.
\PlanarETRINV*

Our proof builds on the work of 
Dobbins, Kleist, Miltzow and Rz{\polhk{a}}{\.{z}}ewski~\cite{AreaUni}.
They showed
that \etrinv is \ER-complete even 
when the variable-constraint incidence graph is planar.
We cannot simply start from their result, because we still need to eliminate equalities, and the obvious idea of replacing $x+y=z$ (for example) by two constraints of the form
$x+y \leq z$ and $x+y \geq z$ will destroy planarity of the   variable-constraint incidence graph. 
\remove{We denote this as \planetrinvOld.
One might be tempted to believe that the following 
reduction from the \planetrinvOld shows 
\planetrinv is \ER-complete. Take an instance $I$ of
\planetrinvOld and replace every addition constraint 
$x+y =z$ by two constraints of the form
$x+y \leq z$ and $x+y \geq z$.
And the same for the inversion gadget. 
The problem with this idea is that the new variable-constraint graph is \emph{not} planar anymore.
} 

\begin{proof}
First note that \planetrinv is in \ER since it is expressible as \ETR{}.  

To prove \ER-hardness we reduce from \etrinv. 
For this purpose let $I$ be an instance of \etrinv.
As first step we replace every equality constraint
by the two corresponding inequality constraints.
For the next step,
let $G_I$ be the variable-constraint graph of $I$.
Let $D$ be a drawing of $G_I$ in the plane with edges
\changed{drawn as straight line segments}
and vertices represented by points. 
We assume that no three segments 
cross in a common point.
This drawing may have crossing edges. 
We will add constraints and variables to $I$ to obtain 
a new instance $J$, which is equivalent to $I$ and such that the corresponding graph $G_J$ is planar. To compute $J$, we replace each crossing in $D$ by a `crossing gadget'. 
Since $D$ has at most a quadratic number of crossings, this construction takes polynomial time.

\begin{figure}[htbp]
    \centering
    \includegraphics[width=\textwidth]{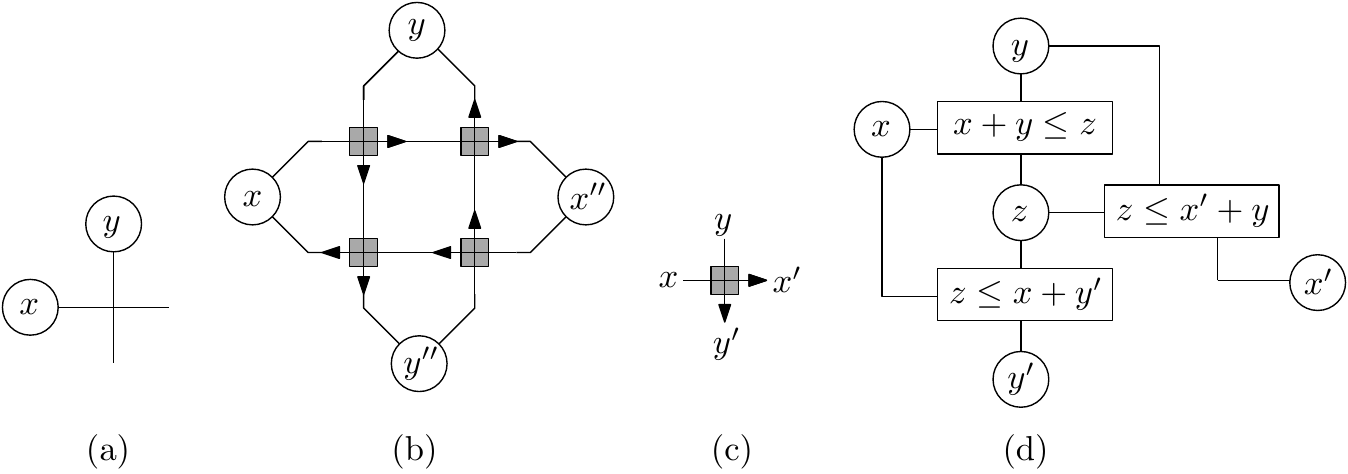}
    \caption{\changed{
    A crossing gadget:  (a) a crossing; (b) a crossing gadget composed of 4 half-crossing gadgets that enforces $x=x''$ and $y=y''$; (c) schematic of a half-crossing gadget; (d) detail of a half-crossing gadget that ensures $x\leq x'$ and $y\leq y'$.
    } 
    }
    \label{fig:halfCrossing}
\end{figure}
We first introduce a \emph{half-crossing gadget} that almost 
does the job, by using an idea of Dobbins et al.~\cite{AreaUni}.
 Figure~\ref{fig:halfCrossing}(d) illustrates this gadget. Note that the inequalities $x+y\le z$ and $z\le x'+y$ ensure that 
$x \leq x'$. Similarly, we can observe that $y\leq y'$. 

To enforce $x=x''$ and $y=y''$
we use four copies of the half-crossing gadget to build a \emph{crossing gadget}.
See 
Figure~\ref{fig:halfCrossing}(b). 
The top two 
half-crossing gadgets (i.e., the pair of gadgets lying on the $x$-monotone path from $x$ to $x''$) ensure $x \leq x''$. 
The bottom two 
half-crossing gadgets ensure $x \geq x''$.
Together they ensure $x=x''$. The same works with $y$ and $y''$.
 
Since the variables of \etrinv are restricted to the range $[1/2,2]$, the variable $z$ will 
lie in the interval $[1,4]$.  
\changed{This is why our definition of \planetrinv involved a larger range for variables.
Accordingly, our final step is to loosen the range restriction of all variables to $[1/2,4]$.
For any new variable $z$ introduced in a half-crossing gadget, the larger range does no harm.
For any original variable $x$, we will enforce the added restriction that $x \le 2$ by adding further constraints.  In particular, add two new variables $a_x$ and $b_x$ and add the constraints $a_x=1$, $b_x=1$, $a_x + b_x \ge x$.  
Note that the variable-constraint incidence graph remains planar, since the new constraints and variables only connect to $x$ in the graph.}

\changed{The final constructed instance $J$ is equivalent to the original, has a planar variable-constraint incidence graph, and restricts all variables to the required range.
}
\remove{
For that purpose assume $A_I$ is an assignment of values
to each variable in $I$ that satisfies all constraints.
This gives us immediately an assignment $A_J$ of all
variables in the new instance $J$
in the range $[1/2,4]$ as required.
The other direction is more tricky. Assume that there
is an assignment $A_J$ of all variables of $J$ in the
required ranges $[1/2,4]$ of all variables 
in $J$ satisfying all constraints of $J$. 
This gives us an
assignment $A_I$ of all variables in $I$. However 
some variables might lie outside the range $[1/2,2]$.
Here a subtlety of the result of 
Abrahamsen et.al.~\cite{AbrahamsenAM18STOC} 
becomes crucial. They showed that for the instance
$I$ either there is no solution at all or there is
a solution in the range $[1/2,2]$.
This is of course not true for any instance, 
but for those instances that resulted from the 
reduction from Abrahamsen et.al.~\cite{AbrahamsenAM18STOC}.
Thus since assignment $A_I$ of the variables
that satisfies all the constraints there must 
be another 
assignment $A_I'$ with all variables in the 
correct range.
}
\qed
\end{proof}



\section{Additional Details of Irrational Coordinates}\label{app:irrational}
This section heavily relies on a paper by 
Abrahamsen et al.~\cite{abrahamsen2017irrational}.
We repeat the key ideas of their paper
and show how to adapt it for our purpose.
In their paper, they studied the \AGP. In the \AGP, we are given a polygon $P$ and a number $k$,
and we want to find a set of at most $k$ guards (points) that together see the entire polygon. We say a guard $g$ sees a point $p$ if the entire line-segment $gp$ is contained inside 
the polygon $P$. Abrahamsen et al. gave a simple polygon with integer coordinates such that there exists only one way to guard it optimally, with three guards. 
Those guards have irrational coordinates. See Figure~\ref{fig:irrationalGuards}, for a sketch of their polygon.

The key ingredients of their proof are as follows. 
First observe that the notches in the polygon boundary force there to be a guard on each of the three so-called \emph{guard segments}, indicated by the dashed lines. Then the left guard and the middle guard together must see the top left pocket edge  and bottom left pocket edge (shown in blue).  Similarly the middle guard and the right guard together must see the top right pocket edge and bottom right pocket edge. Abrahamsen et al.~specify precise coordinates for the polygon that force unique positions for the guards, and such that those positions have irrational coordinates.   As shown in Figure~\ref{fig:irrationalGuards}, the unique guard positions result in a single point on each pocket edge that is seen by two guards.  For example, point $a$ is the only point on the top left pocket edge that is seen by $x$ and $y$.  In particular, the line segments $xa$ and $ya$ pass through reflex corners of the polygon.

We adopt this example as follows, see
Figure~\ref{fig:irrationalVertices}.
Instead of guard segments we use \emph{variable segments} (shown in green), and instead of guards we use vertices.
 We describe variable segments in detail in Section~\ref{sec:MainER},
see also Figure~\ref{fig:segment-end}. 
By placing notches in the polygon boundary with fixed vertices of the graph in the notches, we can force there to be a vertex on each variable segment.
The pocket edges of the previous example become variable segments. 
The middle variable segment, the one that contains vertex $y'$ in the figure, is determined by a hole in the region. 
We need to use a hole in order to keep our graph drawing planar.  Note that, 
besides the fixed vertices lying on the boundary of the region, our graph now has 7 vertices, which are forced to lie on 7 distinct variable segments. 
To complete the construction of our graph, we add edges between the 7 vertices to create two cycles: 
one containing the leftmost four vertices and the other containing the rightmost four vertices, as illustrated by the dotted lines in Figure~\ref{fig:irrationalVertices}. 

Now the constraints on the three vertices $x'$, $y'$ and $z'$, shown with black dots in Figure~\ref{fig:irrationalVertices}, are exactly the same as for the guards $x$, $y$, and $z$ in Figure~\ref{fig:irrationalGuards}. 
All that changes is how the constraints are described. 
Let us give an explicit example. 
The guards $x$ and $y$ together need to see the top left pocket edge. In our new polygon, the vertices $x'$ and $y'$ must both be  adjacent to the same vertex $a'$, as indicated  in Figure~\ref{fig:irrationalVertices}.
This imposes the same constraints on the vertices $x',y'$ as was imposed on the guards $x,y$. This translation of conditions happens in the same way for all the other pockets.

As there exists only one position to guard Abrahamsen et al.'s  polygon with three guards, there exists also only one way to place the vertices in the polygon of 
Figure~\ref{fig:irrationalVertices}, and those positions are irrational.

\section{Additional Details for \ER-completeness}
\label{sec:addition}

\begin{proof}[Proof of Lemma~\ref{lem:addition}]
\anna{Can we make the word "proof" appear?}

    This proof is inspired by the following thought experiment.
    Assuming that we choose $z$ always to be the maximum 
    possible value. Furthermore we assume that while we 
    fix the position of $y$, we move $x$ some distance $d$ to the
    left. What we would expect is that $z$ moves by the same distance
    to the left. Actually, showing the last
    statement also proves the lemma, due to symmetry of $x$ and $y$. 
    We denote by $\ell$ the line that contains the variable segments
    of $x$ and $y$. We denote by $t$ half the distance that
    $z$ moves. Note that $t$ has a geometric interpretation
    as indicated in Figure~\ref{fig:AdditionProof}.
    We need to show $d = 2t$. 
    The lengths $A,A',B,B'$ are defined, by Figure~\ref{fig:AdditionProof}.
    Note that $B' = 2B$, because $\|a-b\| = \|b-c\| $.
    Similarly, follows $A'= 2A$.
    The lemma follows from
    \[d = B'-A'=2(B-A) = 2t.\]
\qed
\end{proof}
\begin{figure}[htb]
    \centering
    \includegraphics[width=.8\textwidth]{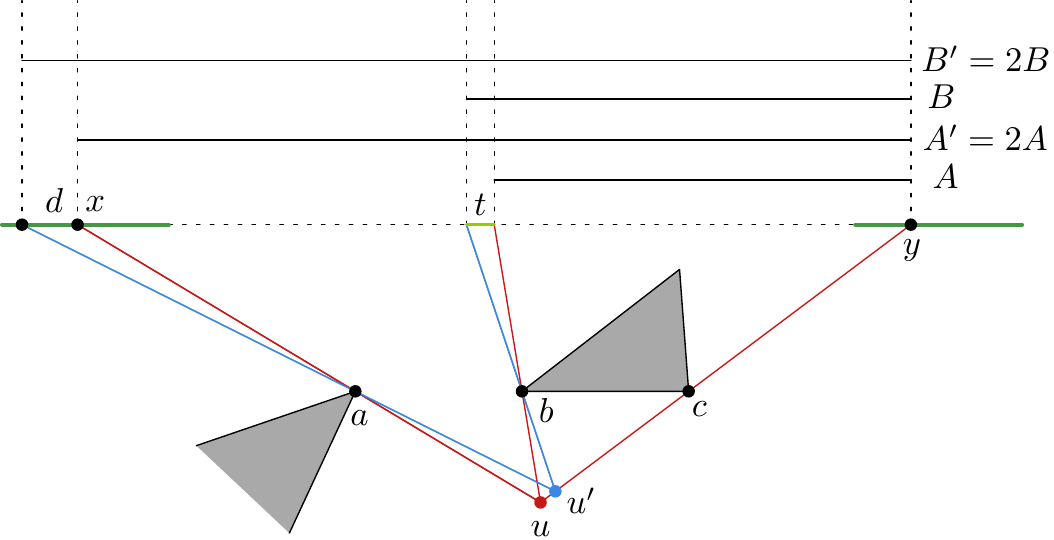}
    \caption{An illustration of the correctness 
    of the addition gadget.}
    \label{fig:AdditionProof}
\end{figure}

\begin{figure}[pt]
    \centering
    \includegraphics[width=.9\textwidth]
    {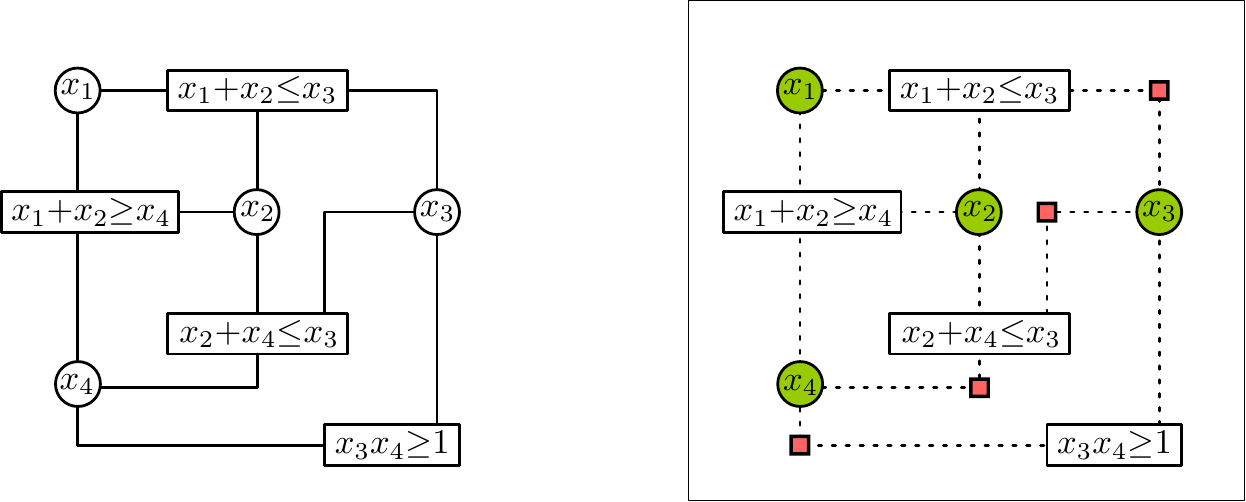}
    \caption{
    \changed{The overall reduction from \planetrinv to \DGP. }
    Left: A rectilinear drawing $D$ of the variable-constraint graph of $I$.
    In this example every variable vertex has degree at most $3$.
    Right:  A schematic representation of the resulting instance. The variable-segments are placed inside
    the variable vertices (green). Every edge is replaced by
    a sequence of copy gadgets (dotted lines) 
    and turn gadgets (red squares). 
    The constraint vertices are replaced by the corresponding
    constraint gadgets. Note that it might be necessary to have one or several turn and copy gadgets as part of the 
    addition and inversion gadgets.
    }
    \label{fig:reduction}
\end{figure}

\end{document}